\documentclass{llncs}

\usepackage{amsmath,amssymb,amsfonts}
\usepackage[textsize=small,textwidth=1.3in]{todonotes}
\usepackage{appendix}

\newcommand{\Keywords}[1]{\par\noindent 
{\small{\bfseries Keywords\/}: #1}}

\newcommand{\MSC}[1]{\par\noindent 
{\small{\bfseries MSC(2010)\/}: #1}}

\title{Error-correcting pairs for a public-key cryptosystem}
\titlerunning{ECP for PKC}

\author{Irene M\'arquez-Corbella\inst{1} \and Ruud Pellikaan\inst{2}}
\authorrunning{I. M\'arquez-Corbella \and R. Pellikaan}
\institute{Department of Algebra, Geometry and Topology, University of Valladolid, Facultad de Ciencias, 47005 Valladolid, Spain. \email{imarquez@agt.uva.es}
\thanks{The first author is partially supported by Spanish MCINN under project MTM2007-64704 and by a FPU grant AP2008-01598 by Spanish MEC.}
\and
Department of Mathematics and Computing Science, Eindhoven University of Technology, P.O. Box 513, 5600 MB Eindhoven, The Netherlands. \email{g.r.pellikaan@tue.nl}}

\begin{document}
\maketitle

\begin{abstract}
Code-based cryptography is an interesting alternative to classic number-theory PKC since it is conjectured to be secure against quantum computer attacks. Many families of codes have been proposed for these cryptosystems, one of the main requirements is having high performance $t$-bounded decoding algorithms which in the case of having high an error-correcting pair is achieved. In this article the class of codes with a $t$-ECP is proposed for the McEliece cryptosystem. The hardness of retrieving the $t$-ECP for a given code is considered. As a first step distinguishers of several subclasses are given.
\end{abstract}

\Keywords{Code-based Cryptography, Error-Correcting Pairs.}
\MSC{11T71, 94A60, 94B05.}
%

\section{Introduction}

The notion of Public Key Cryptography (PKC) was first introduced in $1976$ \cite{diffie:1982} by Diffie and Helman, though Merkle had previously developed some of the key concepts \cite{merkle:1982}. The main advantage with respect to symmetric-key cryptography is that it does not require an initial exchange of secrets between sender and receiver. In the survey paper \cite{koblitz:2010} it is stated that \\

\emph{``At the heart of any public-key cryptosystem is a one-way function - a function $y=f(x)$ that is easy to evaluate but for which is computationally infeasible (one hopes) to find the inverse $x=f^{-1}(y)$".} \\

The most famous trapdoor one-way functions are:
\begin{itemize}
\item \textbf{Integer factorization} where $x=(p,q) $ is a pair of prime numbers and $y=pq$ is its product. The best-known example of PKC is the Rivest-Shamir-Adleman (RSA) cryptosystem whose security is based on the hardness of distinguishing prime numbers from composite number, i.e. the intractability of inverting this one-way function.
\item \textbf{Discrete logarithm} for which a group $G$ (written multiplicatively) and an element $a\in G$ are required, then $x$ is an integer and $y=a^x$. The security of the ElGamal cryptosystem or the Diffie-Hellman key exchange depends on the difficulty of finding discrete logarithms modulo a large prime.
\item \textbf{Elliptic curve discrete logarithm} which it is actually a particular case of the previous function when $G$ is taken as an elliptic curve group. Then $x=P$ is a point on the curve and $y=kP$ is another point on the curve obtained by the multiplication of $P$ with a scalar $k$. Elliptic Curve Cryptography (ECC) proposed independently by Koblitz \cite{koblitz:1987} and Miller \cite{miller:1986} in $1985$  is based on the difficulty of this function in the group of points on an elliptic curve over a finite field.
\end{itemize}

However with the discovery of Shor's algorithm \cite{Shor:1997} anyone with a quantum computer can break in polynomial time all cryptosystems whose security depends on the difficulty of the previous problem. Post-quantum cryptography gave birth to the next generation of cryptography algorithms, which are designed to run on conventional computers but no attacks by classical or quantum computers are known against them. See \cite{bernstein:2009b} for an overview of the state of the art in this area. Code-based cryptosystems such as McEliece \cite{mceliece:1978} and Niederreiter \cite{niederreiter:1986} cryptosystems are interesting candidates for post-quantum cryptography. See the surveys \cite{biswas:2008,engelbert:2007,overbeck:2007,sendrier:2005a,sendrier:2005b}.

The security of code-based cryptosystems is connected to the hardness of the general decoding problem which was shown by  Berlekamp-McEliece-Van Tilborg \cite{barg:1998,berlekamp:1978} to be NP-hard, even if preprocessing is allowed \cite{bruck:1990}. However it is not known whether this problem is almost always or in the average difficult. The problem of \emph{minimum distance decoding} with input $(G, \mathbf y)$ where $G$ is a generator matrix of a code $C$ over $\mathbb F_q$ of parameters $[n,k,d]$ addresses to determine a codeword $\mathbf c \in C$ of minimal distance to $\mathbf y$. The \emph {bounded distance decoding problem} depends on a function $t(n,k,d)$. The input is again $(G, \mathbf y)$ but the output is a codeword $\mathbf c\in C$ (if any) verifying that $d(\mathbf y ,C)\leq t(n,k,d)$, where $d(\cdot, \cdot)$ denotes the hamming distance between two vectors on $\mathbb F_q^n$. Moreover \emph{decoding up to half the minimum distance} is a bounded distance decoding problem such that $t(n,k,d)\leq \lfloor (d-1)/2\rfloor$ for all $n,k$ and $d$.

All known minimum distance decoding algorithm for general codes have exponential complexity in the length of the code. However there are several classes of codes such as the Reed-Solomon, BCH, Goppa or algebraic geometry codes which have polynomial decoding algorithms that correct up to a certain bound which is at most half the minimum distance.

The problems posed above have two parts \cite{hoeholdt:1995}. Firstly the \emph{preprocessing} part done at a laboratory or a factory where for an appropriate code $C$ a decoder ${\cal A}_C$ is  built which is allowed to be time consuming. Secondly the
actual operating of the many copies of the  decoder for consumers which should work very fast. So we can consider the problem of \emph{minimum distance decoding with preprocessing}. From the error-correction point of view it seems pointless to decode a bad code, but for breaking the McEliece cryptosystem one must be able to decode efficiently all, or almost all, codes.

In $1978$ \cite{mceliece:1978} McEliece presents the first PKC based on the theory of error-correcting codes. Its main advantages are its fast encryption and decryption schemes. However the large key size of its public key makes it very difficult to use in many practical situations. In this cryptosystem the \emph{public key space} $\mathcal K$ is the collection of all generator matrices of a chosen class of codes that have an efficient decoding algorithm that corrects all patterns of $t$ errors, the \emph{plaintext space} is ${\cal P}=\mathbb F_q ^k\times W_{n,q,t}$, where $W_{n,q,t}$ is the collection of all $\mathbf e \in \mathbb F_q ^n$ of weight $t$, and the \emph{ciphertext space} is ${\cal C}=\mathbb F_q ^n$. The \emph{sample space } is given by $\Omega ={\cal P}\times {\cal K}$. The \emph{encryption map}
$\begin{array}{cccc}
E_G: &{\cal P} &\rightarrow &{\cal C}
\end{array}$ for a given key $G \in {\cal K}$ is defined by $E_G(\mathbf m,\mathbf e )=\mathbf m G + \mathbf e$. An \emph{adversary} ${\cal A}$ is a map from
${\cal C}\times {\cal K}$ to ${\cal P}$. This adversary is successful for $(x,G) \in \Omega $ if ${\cal A}(E_G(x),G)=x$.

Let $\mathcal C $ be a class of codes such that every code $C$ in $\mathcal C $ has an efficient decoding algorithm correcting all patterns of $t$ errors.
Let $G\in \mathbb F_q^{k\times n}$ be a generator matrix of $C$. In order to mask the origin of $G$, take a $k\times k$ invertible matrix $S$ over $\mathbb F_q$ and an $n \times n$ permutation or monomial matrix $P$. Then for the McEliece PKC the matrices $G$, $S$ and $P$ are kept secret while $G'=SGP$ is public. Furthermore the (trapdoor) one-way function of this cryptosystem is usually presented as follows:
$$x=(\mathbf m,\mathbf e) \mapsto y= \mathbf m G' + \mathbf e,$$
where $\mathbf m \in \mathbb F_q ^k$ is the plaintext and $\mathbf e \in \mathbb F_q ^n$ is a random error vector with hamming weight at most $t$.

McEliece proposed to use the family of Goppa codes. The problem of bounded distance decoding for the class of codes that have the same parameters as the Goppa codes is difficult in the worst-case \cite{finiasz:2004}. However, it is still an open problem whether decoding up to half the minimum distance is NP-hard which is the security basis of the McEliece cryptosystem. Algebraic geometry codes were also proposed for the McEliece PKC in \cite{janwa:1996,niebuhr:2006}. The security of this PKC is based on two assumptions \cite{biswas:2008,kobara:2001}:
\begin{enumerate}
\item[A.1] In the average it is difficult to decode $t$ errors for all codes that have the same parameters as the codes used as key,
\item[A.2] It is difficult to distinguish arbitrary codes from those coming from ${\cal K}$.
\end{enumerate}

Concerning the first assumption it might be that the class of codes is too small or too rigid.
For instance Sidelnikov-Shestakov \cite{sidelnikov:1992} gave an adversary that is always successful if one take for public key space the generator matrices of generalized Reed-Solomon (GRS) codes. Concerning the second assumption recent progress is made by Faug\`ere et al. \cite{faugere:2011,otmani:2011} where they showed that one could distinguish between high rate Goppa, alternant and random codes.

In $1986$ \cite{niederreiter:1986} Niederreiter presented a dual version of McEliece cryptosystem which is equivalent in terms of security \cite{li:1994}. Niederreiter's system differs from McEliece's system in the public-key structure (it use a parity check matrix instead of a generator matrix of the code), in the encryption mechanism (we compute the syndrome of a message by the public key) and in the decryption message. In its original paper Niederreiter proposed the class of GRS codes.

Let $H\in \mathbb F_q^{(n-k)\times n}$ be a parity check matrix of a code $C$ in $\mathcal C $. $H$ is masked by $H'=SHP$, where $S$ is an invertible matrix over $\mathbb F_q$ of size $n-k$ and $P$ is an $n \times n$ permutation or monomial matrix. The (trapdoor) one-way function in case of the Niederreiter PKC is presented by
$$x=\mathbf e \mapsto y= \mathbf e H'^T,$$
where $\mathbf e \in \mathbb F_q ^n$ has weight $t$.

In a \emph{syndrome based (SB) hash} function \cite{augot:2005,finiasz:2007,finiasz:2008} an $n \times r$ parity check matrix $H$ is chosen at random, then SB hash system is given by a procedure that encodes $\mathbf s$ bits of information into a word $\mathbf e$ of length $n$ and weight $t$. The one-way function in this case (which has no trapdoor) is given by
$$x=\mathbf e \mapsto y= \mathbf e H^T.$$

It was shown in \cite{duursma:1994,pellikaan:1992,pellikaan:1996} that the known efficient bounded distance decoding algorithms of Reed-Solomon, BCH, Goppa and algebraic geometry codes can be described by a basic algorithm using an error correcting pair. That means that the proposed McEliece cryptosystem are not based on the inherent tractability of bounded distance decoding but on the one-way function
$$x=(A,B) \mapsto y=A*B,$$
where $(A,B)$ is a $t$-error-correcting pair.

Consider $\mathcal C_t$, the class of linear codes over $\mathbb F_q$  that have a $t$-error correcting pair over an extension of $\mathbb F_q$. It was shown by Pellikaan \cite{pellikaan:1992} that codes of this family have an efficient decoding algorithm that corrects up to $t$ errors. This  makes them appropriate for code-based cryptography. Note that most families of codes used in such cryptosystems belong to $\mathcal C_t$ such as the generalized Reed-Solomon codes, the Goppa codes, the alternant codes and the algebraic-geometry codes.

For further details on the notion of error-correcting pair see Section \ref{section2} where we formally review this definition and we give a brief survey on the properties that are relevant to this work.

The aim of this paper is to study the subclass of $\mathcal C_t$ formed by those linear codes $C$ whose error correcting pair is not easily reconstructed from $C$ . Section \ref{section3} deals with the security status of this scheme, detailing the state-of-art and the existence of error-correcting pairs for families of codes most commonly used in code-based cryptography.

Finally, in Section \ref{section4}, following the work of Faug\`ere et al. \cite{faugere:2011}, we present distinguishers for several families of codes. Recall that the hardness of the distinguishing problem was part of the basis of the security of code-based cryptosystems.

\section{Error-correcting pairs}
\label{section2}

From now on the dimension of a linear code $C$ will be denoted by $k(C)$ and its minimum distance by $d(C)$. Given two elements $\mathbf a$ and $\mathbf b$ on $\mathbb F_q^n$, the \emph{star multiplication} is defined by coordinatewise multiplication, that is
$\mathbf a * \mathbf b = (a_1b_1, \ldots, a_nb_n)$
while the  \emph{standard inner multiplication} is defined by $ \mathbf a \cdot \mathbf b = \sum_{i=1}^n a_i b_i$.
In general, for two subsets $A$ and $B$ of $\mathbb F_q^n$ the set $A*B$ is given by
$\{\mathbf a * \mathbf b \mid \mathbf a \in A \hbox{ and } \mathbf b \in B \}$. Furthermore $A\perp B$ if and only if $\mathbf a \cdot \mathbf b = 0$ for all $\mathbf a \in A$ and $\mathbf b\in B$.

\begin{definition}\label{ECP}
Let $C$ be a linear code in $\mathbb F_q^n$. The pair $(A, B)$ of $\mathbb F_{q^m}^n$ is called a \emph{t-error correcting pair} (ECP) for $C$ if the following properties hold:
\begin{enumerate}
\item[E.1] $(A*B) \perp C$,
\item[E.2] $K(A) > t$,
\item[E.3] $d(B^{\perp}) > t$,
\item[E.4] $d(A) + d(C) > n$.
\end{enumerate}
\end{definition}

The notion of an error-correcting pair for a linear code was introduced independently by Pellikaan in \cite{pellikaan:1988} and K\"otter in \cite{koetter:1992}. In \cite{pellikaan:1988} it is shown that a linear code in $\mathbb F_q^n$ with a $t$-error correcting pair has a decoding algorithm which corrects up to $t$ errors with complexity $\mathcal O (n^3)$. Furthermore the minimum distance of such linear code is at least $2t+1$.

The existence of ECP for generalized Reed-Solomon and Algebraic codes was shown in \cite{pellikaan:1988} and for many cyclic codes Duursma and K\"otter in \cite{duursma:1994} have found ECP which correct beyond the designed BCH capacity.

Note that if E.4 is replaced by the following statements
\emph{
\begin{enumerate}
\item[E.5] $d(A^{\perp}) > 1$ i.e. $A$ is a non-degenerated code,
\item[E.6] $d(A) + 2t > n$.
\end{enumerate}}
then $d(C)\geq 2t+1$ and $(A,B)$ is a $t$-ECP for $C$.

\section{Error-correcting pairs for public key cryptosystems}
\label{section3}

Let $\mathcal P_t$ be the collection of pairs $(A, B)$ such that $A$, $B$ are linear codes over some extension of $\mathbb F_q$, $A$ is non-degenerated and $(A,B)$ is a $t$-error correcting pair for some linear code $C$ in $\mathbb F_q^n$. We consider the following one way function
$$\begin{array}{cccc}
\varphi: & \mathcal P_t & \longrightarrow & \mathbb F_q^n \\
& x = (A, B) & \longmapsto & y = A*B
\end{array}$$

Let $U$ and $V$ be two generator matrices with rows denoted by $\mathbf u_i$ and $\mathbf v_i$, respectively, $U*V$ be the matrix form by the rows $\mathbf u_i * \mathbf v_j$ ordered lexicographically and $\mathrm{red}(U*V)$ be the matrix obtained from $U*V$ by deleting dependent rows. Then the implementation of $\varphi$ may be given by
$$\begin{array}{ccc}(U,V) \longmapsto & y = \mathrm{red} (U*V)\end{array}$$

Firstly we note that $\mathbf uP * \mathbf vP = (\mathbf u*\mathbf v)P$ for every permutation or monomial matrix $P$. Thus, if $(A,B)$ is a $t$-ECP for $C$, then $(AP, BP)$ is a $t$-ECP for $P^{-1}C$. Furthermore, let $S_1$ and $S_2$ be invertible matrices of the correct sizes to be multiplied by the matrices $U$ and $V$, respectively, then $U*V$ generates the same code as $(S_1U) * (S_2V)$ since $(S_1U)*\mathbf v = S_1(U*\mathbf v)$ and $\mathbf u * (S_2V) = S_2(\mathbf u * V)$ for all vectors $\mathbf u$ and $\mathbf v$. Therefore the masking $H'=SHP$ by means of an invertible matrix $S$ and a permutation matrix $P$ is already incorporated in the choice of the pair of generator matrices $(U,V)$.

Let $C$ be the code with the elements of $A*B$ as parity checks. If the one-way function $\varphi$ is indeed difficult to invert,
then the code $C$ with parity check matrix $H=\mbox{red}(U*V)$ might be used as a public-key in a coding based PKC.
Otherwise it would mean that the PKC based on codes that can be decoded
by error-correcting pairs is not secure. In the following we consider seven collections of pairs.

\begin{example}
\label{example1}
The class of GRS codes was proposed for code-based PKC by Niederreiter \cite{niederreiter:1986}. However this proposal is completely broken by the Sidelnikov-Shestakov attack given in \cite{sidelnikov:1992}.\\
Let $\mathbf a$ be an $n$-tuple of mutually distinct elements of $\mathbb F_q$ and $\mathbb b$ be an $n$-tuple of nonzero elements of $\mathbb F_q$. Then the \emph{generalized Reed-Solomon} code $\mathrm{GRS}_k(\mathbf a, \mathbf b)$ is defined by
$$\mathrm{GRS}_k(\mathbf a, \mathbf b) = \left\{ \left(f(a_1)b_1, \ldots , f(a_n)b_n\right)\mid f(X) \in \mathbb F_q[X] \hbox{ and } \deg(f(X)) < k\right\}.$$
That is, if we define by induction $\mathbf a^1=\mathbf a$ and $\mathbf a^{i+1}=\mathbf a *\mathbf a^i$, then $\mathrm{GRS}_k(\mathbf a, \mathbf b)$ is generated by the elements $\mathbf b * \mathbf a^i$ with $i=0, \ldots, k-1$, i.e. if $k\leq n\leq q$, then $\mathrm{GRS}_k(\mathbf a ,\mathbf b )$ is an $[n,k,n-k+1]$ code. Furthermore the dual of a GRS code is again a GRS code, in particular $\mathrm{GRS}_k(\mathbf a, \mathbf b)^{\perp} = \mathrm{GRS}_{n-k}(\mathbf a, \mathbf b')$ for some $\mathbf b'$ that is explicitly known.\\
Let $A = \mathrm{GRS}_{t+1}(\mathbf a, \mathbf u)$, $B = \mathrm{GRS}_t(\mathbf a, \mathbf v)$ and $C = \mathrm{GRS}_{2t}(\mathbf a, \mathbf u*\mathbf v)^{\perp}$. Then $(A,B)$ is a $t$-ECP for $C$.
Conversely let $C = \mathrm{GRS}_{k}(\mathbf a, \mathbf b)$, then $A=\mathrm{GRS}_{t+1}(\mathbf a, \mathbf b')$ and $B=\mathrm{GRS}_{t}(\mathbf a, \mathbf 1)$ is a $t$-ECP for $C$ where $t=\left\lfloor\frac{n-k}{2}\right\rfloor$ and $\mathbf b'\in \mathbb F_q^n$ is a nonzero vector verifying that $\mathrm{GRS}_{k}(\mathbf a, \mathbf b)^{\perp} = \mathrm{GRS}_{n-k}(\mathbf a, \mathbf b')$.\\
So GRS codes are the prime examples of codes that have a $t$-error-correcting pair.
Moreover if $C$ is an $[n,n-2t,2t+1]$ code which has a $t$-error-correcting pair, then $C$ is a generalized Reed-Solomon.
This is trivial if $t=1$, proved for $t=2$ in \cite[Theorem 6.5]{pellikaan:1996} and for arbitrary $t$ in \cite{marquez:2012b}.
\end{example}

\begin{example}\label{example2}Error-correcting pairs for cyclic codes were found by Duursma and K\"otter \cite{duursma:1993a,duursma:1994,koetter:1996}.
Cyclic codes are not considered for applications in code-based PKC.
\end{example}

\begin{example}
\label{example3}
The class of subcodes of GRS codes was proposed by Berger-Loidreau \cite{berger:2005} for code-based PKC to resist precisely the Sidelnikov-Shestakov attack. But for certain parameter choices this proposal is also not secure as shown by  Wieschebrink \cite{wieschebrink:2006,wieschebrink:2010} and M\'arquez et al. \cite{marquez:2011a}.\\
Let $C$ be a subcode of the code $\mathrm{GRS}_{n-2t}(\mathbf a, \mathbf b)$. This GRS code has a $t$-error-correcting pair by Example \ref{example1} which is also a $t$-ECP for $C$.
\end{example}

\begin{example}
\label{example4}
Goppa codes were proposed for McEliece PKC by its author \cite{mceliece:1978}. Sidelnikov-Shestakov made a claim \cite{sidelnikov:1992} that their method for GRS  codes could be extended to attack Goppa codes as well,
but this was never substantiated by a paper in the public domain.
In its original paper McEliece recommend the class of binary Goppa codes with parameters $[1024,524,101]$, but this proposal is no longer
secure with nowadays computing power as shown in Peters et al. \cite{bernstein:2008,peters:2010,peters:2011}
by improving decoding algorithms for general codes. The attack of Wieschebrink \cite{wieschebrink:2010} is not yet efficient enough to be applicable to these codes.\\
A Goppa code associated to a Goppa polynomial of degree $r$ can be viewed as an alternant code, that is a subfield subcode of GRS code of codimension
$r$ and therefore they have also a $\lfloor r/2 \rfloor $-error-correcting pair.
In the binary case with an associated square free polynomial the Goppa code has an $r$-ECP. 
\end{example}

\begin{example}
\label{example5}
Algebraic geometry (AG)  codes were introduced in 1977 by V.D. Goppa and were proposed by Janwa-Moreno \cite{janwa:1996} for the McEliece PKC. Recall that GRS codes can be seen as the class of AG codes on the projective line, i.e. the algebraic curve of genus zero. We refer the interested reader to \cite{goppa:1977,stichtenoth:1993,tsfasman:1991}.\\
Let ${\cal X} $ be an algebraic curve defined over ${\mathbb F_q}$ with genus $g$. By an algebraic curve we mean a curve that is absolutely irreducible, nonsingular and projective.
Let ${\cal P} $ be an $n$-tuple of ${\mathbb F_q} $-rational points on ${\cal X}$ and let $E$ be a divisor of ${\cal X} $ with disjoint support from ${\cal P}$ of degree $m$. Then the algebraic geometry code $C_L({\cal X},{\cal P},E)$ is the image of the Riemann-Roch space $L(E)$ of rational functions with prescribed behavior of zeros and poles at $E$ under the evaluation map $\mbox{ev}_{{\cal P}} $. If $m<n$, then the dimension of the code $C_L({\cal X},{\cal P},E)$ is at least $ m+1-g$ and its minimum distance is at least $n-m$. If $m>2g-2$, then its dimension is $m+1-g$.
The dual code $C_L({\cal X},{\cal P},E)^{\perp}$ is again AG.
If $m>2g-2 $, then the dimension of the code $C_L({\cal X},{\cal P},E)^{\perp}$ is at least $n-m-1+g$ and its minimum distance is at least $d^*=m-2g+2$. If $m<n$, then its dimension is $n-m-1+g$.\\
If $A=C_L({\cal X},{\cal P},E)$ and $B=C_L({\cal X},{\cal P},F)$, then $\langle A*B \rangle \subseteq C_L({\cal X},{\cal P},E+F)$.
So there are abundant ways to construct error-correcting pairs of an AG code.
An AG code on a curve of genus $g$ with designed minimum distance $d^*$
has a $t$-ECP over ${\mathbb F_q} $ with $t=\lfloor (d^*-1-g)/2 \rfloor$ by \cite[Theorem 1]{pellikaan:1989} and \cite[Theorem 3.3]{pellikaan:1992}.
If $e$ is sufficiently large, then there exists a $t$-ECP over ${\mathbb F_{q^e}} $ with $t=\lfloor (d^*-1)/2 \rfloor$ by \cite[Proposition 4.2]{pellikaan:1996}.\\
It was shown by M\'arquez et al. \cite{marquez:2011b,marquez:2012a} that these codes are not secure for rates $R$ in the intervals
$[\gamma , \frac{1}{2}-\gamma ]$, $[\frac{1}{2}+\gamma  , 1-\gamma]$, $[\frac{1}{2}-\gamma ,  1-3\gamma]$ and $[3\gamma ,\frac{1}{2}+\gamma]$,
where $R=k/n$ is the information rate and $\gamma =g/n$ the relative genus.\\
Geometric Goppa codes, which are subfield subcodes of Algebraic geometry codes \cite{skorobogatov:1991} generalizing the classical Goppa codes that are subfield subcodes of GRS codes, were proposed for the McEliece PKC by Janwa-Moreno \cite{janwa:1996}.
\end{example}

\begin{example}
\label{example6}
If $(A,B)$ is a pair of codes with parameters $[n,t+1,n-t]$ and $[n,t,n-t+1]$, respectively,
and $C =(A*B)^{\perp}$, then the minimum distance of $C$ is at least $2t+1$ and $(A,B)$ is a $t$-error-correcting pair for $C$
by \cite[Corollary 3.4]{pellikaan:1996}. The dimension of $\langle A*B \rangle $ is at most $t(t+1)$. So the dimension of $C$ is at least $n-t(t+1)$. In Appendix A it will be shown that this is almost always equal to $n-t(t+1)$ for random choices of $A$ and $B$.\\
If $q$ is considerably larger than $n$, then a random code is MDS. So taking random
codes $A$ and $B$ of length $n$ and dimensions $t+1$ and $t$, respectively, gives a very large class of code for the McEliece PKC. However with large field the key size becomes larger and recall that the main obstacle for coded-based crypto systems was the key size.
\end{example}

\begin{example}\label{example7}If  $(A,B)$ is a pair of codes that satisfy the conditions (E.1), (E.2), (E.3), (E.5) and (E.6),
then the minimum distance of $C$ is at least $2t+1$ and $(A,B)$ is a $t$-error-correcting pair for $C$ by \cite[Corollary 3.4]{pellikaan:1996}.
\end{example}

\section{Distinguishing a code with an ECP}
\label{section4}

Let $\mathcal K$ be a collection of generator matrices of codes that have a $t$-error-correcting pair and that is used for a coded-based PKC system. In this section we address assumption A.2 whether we can distinguish arbitrary codes from those coming from $\mathcal K$.

Let $C$ be a $k$ dimensional subspace of $\mathbb F_q^n$ with basis $\mathbf g_1, \ldots, \mathbf g_k$ which represents the rows of the generator matrix $G \in \mathbb F_q^{k\times n}$.
We denote by $S^2(C)$ the \emph{second symmetric power} of $C$, or equivalently the \emph{symmetrized tensor product} of $C$ with itself. If $\mathbf x_i=\mathbf g_i$, then $S^2(C)$ has basis $\left\{ \mathbf x_i\mathbf x_j \mid 1\leq i\leq j \leq n\right\}$ and dimension $\binom{k+1}{2}$. Furthermore we denote by $\langle C* C\rangle$ or $C^{(2)}$  the \emph{square} of $C$, that is the linear subspace in $\mathbb F_q^n$ generated by $\left\{ {\bf a} *{\bf b} | {\bf a} ,{\bf b} \in C \right\}$.  See \cite[\S 4 Definition 6]{cascudo:2009} and \cite{marquez:2011a,wieschebrink:2010}. Now, following the same scheme as in \cite{marquez:2011b},
we consider the linear map
$$
\begin{array}
{cccc}\sigma : & S^2(C) & \longrightarrow & C^{(2)},
\end{array}
$$
where the element $\mathbf x_i\mathbf x_j$ is mapped to ${\bf g}_i*{\bf g}_j$. The kernel of this map will be denoted by $K^2(C)$.
Then $K^2(C)$ is the solution space of the following set of equations:
$$
\sum_{1\leq  i \leq i' \leq k} g_{ij}g_{i'j} \mathbf X_{ii'} = 0 , \ \ 1 \leq j  \leq n.
$$
There is no loss of generality in assuming $G$ to be systematic at the first $k$ position, making a suitable permutation of columns and applying Gaussian elimination, if necessary. Then $G=\left(\begin{array}{cc}I_k & P\end{array}\right)$ where $I_k$ is the $k\times k$ identity matrix and $P$ is an $k\times (n-k)$ matrix formed by the last $n-k$ columns of $G$. Now $H=\left( \begin{array}{cc}P^T & -I_{n-k}\end{array}\right)$ is a parity check matrix of $C$, or equivalently $H$ is a generator matrix of the $[n,n-k]$ code $D = C^{\perp}$.

In \cite[\S 3]{faugere:2011} and \cite[Ch. 10]{otmani:2011} a system $\mathcal L_{P}$ associated to the matrix $P$ of $k$ linear equations involving the $\binom{n-k}{2}$ variables $ Z_{jl}$, with $k+1 \leq j < l \leq n$, is defined as

$$\mathcal L_{P} = \left\{ \sum_{k< j < j' \leq n} p_{ij}p_{ij'} \mathbf Z_{jj'} = 0  \mid 1 \leq i  \leq k. \right\}$$

This system differs from the system of equations obtained for the kernel $K^2(C)$ in interchanging indices $i$ and $j$ and the strict inequality $j<j'$ in the summation, instead of $i\leq i'$. Denote the kernel of $\mathcal L_P $, that is the space of all solutions of $\mathcal L_P$, by $K(\mathcal L_P )$.

\begin{proposition}\label{p-two-kernels}
$$
\dim K(\mathcal L_P ) = \dim K^2(D)
$$
\end{proposition}

\begin{proof}
Let $M$ be the $\binom{k+1}{2}\times n$ matrix with entries
$\left( g_{ij}g_{i'j}\right)_
{\begin{subarray}{l}
        1 \leq i \leq i' \leq k\\  1\leq j \leq n
\end{subarray}}$. Then a basis of $K^2(C)$ can be read of directly as the kernel of $M$. Note also that the dimension of $C^{(2)}$ is equal to the rank of $M$. Furthermore, since $C^{(2)}$ is the image of the linear map $\sigma$, by the first isomorphism theorem we get:
$$\dim K^2(C) + \dim C^{(2)} = \dim S^{2}(C) = \binom{k+1}{2}.$$

Let $\mathbf h_i$ be the $i$-th row of the parity check matrix $H$, $\mathbf e_i$ be the $i$-th vector in the canonical basis of $\mathbb F_q^{n-k}$ and $\mathbf q_i$ be the $i$-th row of the matrix $P^T$. Then $q_{ij}=p_{j,i+k}$ and ${\bf h}_i=({\bf q}_i|-{\bf e}_i)$. Therefore
$$
\mathbf h_j*\mathbf h_{j'}=
\left\{
\begin{array}{ll}
\left(\begin{array}{c|c}{\mathbf q}_j*{\mathbf q}_j & {\mathbf e}_j\end{array}\right)&\mbox{if } j=j',\\
\left(\begin{array}{c|c}{\mathbf q}_j*{\mathbf q}_{j'} & {\mathbf 0} \end{array}\right)  &\mbox{if } j<j'.
\end{array}
\right.
$$

Let $M_1$ be the $k\times \binom{n-k}{2}$ matrix with entries $\left( p_{ij}p_{ij'}\right)_
{\begin{subarray}{l}
        1 \leq i \leq k\\  k< j< j' \leq n
\end{subarray}}$, then

$$\dim K(\mathcal L_P ) = \binom{n-k}{2} - \mathrm{rank} (M_1)$$

Now let $M_2$ be the $\binom{n-k+1}{2} \times n$ matrix with entries
$\left( h_{ij}h_{i'j}\right)_
{\begin{subarray}{l}
        1 \leq i \leq i' \leq n-k\\  1\leq j \leq n
\end{subarray}}$. Then

$$\dim D^{(2)} = \mathrm{rank} (M_2) =  n-k + \mathrm{rank} (M_1)$$

Therefore
\begin{eqnarray*}
\dim K(\mathcal L_P ) &=& \binom{n-k}{2} - \mathrm{rank} (M_1) \\
& = & \binom{n-k}{2}+ n-k - \dim D^{(2)}\\
& = & \dim K^2(D)
\end{eqnarray*} \qed
\end{proof}

The dual statement of Proposition \ref{p-two-kernels} gives: $\dim K(\mathcal L_{P^T} ) = \dim K^2(C)$.

For every $[n,k]$ code $C$ over $\mathbb F_q$ the following inequality holds:
$$\dim C^{(2)} \leq  \min \{ n, \textstyle\binom{k+1}{2} \}.$$
However if the entries of the matrix $P$ are taken independently and identically distributed,
 then the inequality holds with equality with high probability what is actually proved in the next proposition.

\begin{proposition}\label{p-C2}
Let $C$ be an $[n,k]$ code with $n>\binom{k+1}{2}$ chosen at random. Then
$$
\Pr\left(\dim (C_{\tiny \hbox{random}}^{(2)}) =  \binom{k+1}{2} \right) = 1
$$
\end{proposition}

\begin{proof}
Let $C$ be a linear code with parameters $[n, k]$ over $\mathbb F_q$ with $n>\binom{k+1}{2}$.

We have seen in the proof of Proposition \ref{p-two-kernels}, with the role of $C$ and $D=C^\perp $  interchanged that the linear system $\mathcal L_{P^T}$ associated with $C$ consists of $n-k$ linear equations and $\binom{k}{2}$ unknowns. In case $n-k > \binom{k}{2}$ or equivalently $n > \binom{k+1}{2}$ Faug\`{e}re et al. \cite{faugere:2011} proved that the dimension of the solution space of $\mathcal L_{P^T}$ is $0$ with high probability. Therefore under the same hypothesis we have that the dimension of $C_{\mathrm{random}}^{(2)}$ is $\binom{k+1}{2}$ with high probability.\qed
\end{proof}

\begin{example}
Let $C$ be a GRS code with parameters $[n,k]$, take for instance $C = \mathrm{GRS}_k(\mathbf a, \mathbf b)$ where $\mathbf a$ is an $n$-tuple of mutually distinct elements of $\mathbb F_q$ and $\mathbf b$ is an $n$-tuple of nonzero elements of $\mathbb F_q$. Then $C^{(2)}$ is the code $\mathrm{GRS}_{2k-1}(\mathbf a, \mathbf b* \mathbf b)$ if $2k-1 \leq n$ and $\mathbb F_q^n$ otherwise.
Hence $\dim C^{(2)} = \min \{2k-1, n\}$. Therefore
$$\dim K^2(C) = \binom{k+1}{2}-(2k-1) = \binom{k-1}{2}\hbox{ if }2k-1 \leq n.$$
\end{example}

\begin{example}
Let $C$ be a $k$-dimensional subcode of the code $\mathrm{GRS}_l(\mathbf a, \mathbf b)$. Then $C^{(2)}$ is a subcode of the code $\mathrm{GRS}_{2l-1}(\mathbf a, \mathbf b*\mathbf b)$, if $2l-1\leq n$. Thus $$\dim C^{(2)} \leq \min \{2l-1, n\}.$$
Moreover if $4l-3k-1<q$ and $2l-1\leq \binom{k+1}{2}$, then it was shown in \cite{marquez:2011a} that $C^{(2)}$ is equal to $\mathrm{GRS}_{2l-1}(\mathbf a, \mathbf b*\mathbf b)$ with high probability so, under this hypothesis,
$$\Pr \left( \dim C^{(2)} = 2l-1\right) = 1- o(1).$$
The dual code $D = C^{\perp}$ contains the code $\mathrm{GRS}_l(\mathbf a,\mathbf b)^{\perp} = \mathrm{GRS}_{n-l}(\mathbf a, \mathbf b')$. That is, $D^{(2)}$ contains the square of $\mathrm{GRS}_{n-l}(\mathbf a, \mathbf b')$ which is equal to $\mathrm{GRS}_{2n-2l-1}(\mathbf a, \mathbf b'*\mathbf b')$ if $2n-2l-1 \leq n$, or equivalently if $n\leq 2l+1$.
Recall that the star product of the rows of a generator matrix of any linear code gives a generating set for its square code, that is the square of any $[n,s]$ linear code is generated by $\binom{s+1}{2}$ elements. In particular $D^{(2)}$ is generated by  $\binom{n-k+1}{2}$ elements but since $\mathrm{GRS}_{2n-2l-1}(\mathbf a, \mathbf b'*\mathbf b') \subseteq D^{(2)}$ there are at least $\binom{n-l+1}{2} - (2n-2l+1)$ dependent elements of this generating set. Thus
$$\dim D^{(2)} \leq \binom{n-k+1}{2} - \binom{n-l+1}{2} + 2n-2l-1 = \binom{n-k+1}{2} -\binom{n-l-1}{2}.$$
\end{example}

\begin{example}
The problem of distinguishing Goppa, alternant and random codes from each other was studied by Faug\`ere et al. in \cite{faugere:2011}. 
Their experimental results give rise to a conjecture on the dimension of $K(\mathcal L_P )$ for Goppa and alternant codes  of high rate. 
\end{example}

\begin{example}Let $C = \mathcal C_L(\mathcal X, \mathcal P, E)$ where $\mathcal X$ is an algebraic curve over $\mathbb F_q$ of genus $g$, $\mathcal P$ is an $n$-tuple of mutually distinct $\mathbb F_q$-rational points of $\mathcal X$ and $E$ is a divisor of $\mathcal X$ with disjoint support from $\mathcal P$ of degree $m$. Then $C^{(2)} \subseteq  \mathcal C_L(\mathcal X, \mathcal P, 2E)$.\\
Assume moreover that $2g-2<m<n/2$. Then $C$ has dimension $k=m+1-g$ and $C_L(\mathcal X, \mathcal P, 2E)$  has dimension $2m+1-g=k+m$.
Hence $\dim C^{(2)} \leq k+m$. \\
Let $G$ be a generator matrix of an algebraic geometry code $C $. Take the columns of $G$ as homogeneous coordinates of points in $\mathbb P^{m-g}$, this gives a projective system $\mathcal Q = (Q_1, \ldots, Q_n)$ of points in the projective space $\mathbb P^{m-g}(\mathbb F_q)$. Since $m>2g$ there exists an embedding of the curve $\mathcal X$ in $\mathbb P^{m-g}$ of degree $m$
$$\begin{array}{cccc}
\varphi_E: & \mathcal X & \longrightarrow & \mathbb P^{m-g}\\
& P & \longmapsto & \varphi_E(P) = \left(f_0(P), \ldots, f_{m-g}(P)\right)
\end{array}$$
where $\{f_0, \ldots, f_{m-g}\}$ is a basis of $L(E)$ such that $\mathcal Q = \varphi_E(\mathcal P)$ lies on the curve $\mathcal Y = \varphi_E(\mathcal X)$.
The space $I_2(\mathcal Q)$ of quadratic polynomials that vanish on $\mathcal Q$ can be identified with $K^2(C)$.
Furthermore if $2g+2\leq m < \frac{1}{2}n$, then $I_2(\mathcal Y) = I_2(Q)$ and $I(\mathcal Y)$, the vanishing ideal of $\mathcal Y$, is generated by $I_2(Q)$.
Now
$$
\dim K^2(C)  = \binom{k+1}{2} - \dim C^{(2)} \geq  \binom{k}{2} -m.
$$
Therefore $\mathcal Y$ is given as the intersection of at least $\binom{k}{2} -m$ quadrics in  $\mathbb P^{m-g}$.
For more details we refer the reader to \cite{marquez:2011b}.
\end{example}

\begin{example}Let  $t(t+1)<n$.
Let $(A,B)$ be a pair of random codes of dimension $t+1$ and $t$, respectively.
Take $C =(A*B)^\perp$ as in Example \ref{example6}. Then $D=C^\perp =\langle A*B \rangle $.
So $D^{(2)} = \langle A^{(2)}*B^{(2)} \rangle $. Hence 
$$
\dim D^{(2)} \leq \binom{t+2}{2}\binom{t+1}{2}
$$
which is about half the expected $\binom{t(t+1)}{2}$ in case $\binom{t(t+1)}{2}<n$ by Proposition \ref{p-C2},
since $D$ has dimension $t(t+1)$ with high probability by Appendix A.
\end{example}

\bibliographystyle{splncs03}
\bibliography{MP-PKC2012}

\appendix
\section{The dimension of $\langle A*B\rangle$}
Let $A$ and $B$ be two linear codes over $\mathbb F_q$ with parameters $[n,s]$ and $[n,t]$, generated by the set $\{\mathbf a_1, \ldots , \mathbf a_{s}\}$ and $\{\mathbf b_1, \ldots, \mathbf b_t\}$ of vectors in $\mathbb F_q^n$, respectively. Let $M$ be an $st\times n$ matrix over $\mathbb F_q$ whose rows consist on the vectors $\mathbf a_i * \mathbf b_j = (a_{i,1}b_{j,1}, \ldots, a_{i,n}b_{j,n})\in \mathbb F_q^n$ with $i\in\{1, \ldots, s\}$ and $j\in\{1, \ldots, t\}$ ordered lexicographically. Then the rows of $M$ form a generating set of the code $A*B$.

Indeed $M$ is a block-matrix consisting of $s$ blocks $M_i=\left( \mathbf a_i * \mathbf b_j\right)_{1 \leq j \leq t}$ with $i \in \{1, \ldots, s\}$ of size $t\times n$. We define the support of a codeword $\mathbf c=(c_1, \ldots, c_n)$ by $\mathrm{supp}(\mathbf c) = \{i \mid c_i \neq 0\}$. Note that if $i\not\in \mathrm{supp}(\mathbf a_j)$ then the $i$-th column of $M_j$ consists on zeros.

In the following lines, assuming that $st < n$ we will prove that $M$ has full rank with high probability. We proceed by a similar procedure as in Appendix B of \cite{faugere:2011} where it is proved that the solution space of the linear system associated to an arbitrary random linear code is zero with high probability.

Let $E_1=\mathrm{supp}(\mathbf a_1)$. Suppose $|E_1|\geq t$. Let $F_1$ be a subset of $E_1$ with cardinality $t$. To simplify notation and without loss of generality, we can always assume that $F_1$ corresponds to the first $t$ elements in $\mathbf a_1$, by permuting the elements if necessary. Let $M^{(1)}$ be a square submatrix of $M_1$ formed by its first $t$ columns, i.e.
$$M^{(1)} = \left( a_{1,j}b_{i,j}\right)_
{\begin{subarray}{l}
        j \in F_1\\  1 \leq i \leq t
\end{subarray}} \in \mathbb F_q^{t\times t}$$

Now we define by induction $E_i := \mathrm{supp}(\mathbf a_i)\setminus F_{i-1}$ and the subset $F_i$ as the first $t$ elements of the subset, assuming that
$|E_i|\geq t$. The square matrix $M^{(i)} \in \mathbb F_q^{t\times t}$ is obtained from $M_i$ by taking the $F_i$-indexed columns, for $i\in \{1, \ldots, s\}$. Then clearly the following Lemma holds.

\begin{lemma}
If $|E_i|\geq t$ for all $i \in \{1, \ldots, s\}$ then
$$\mathrm{rank}(M) \geq \sum_{i=1}^{s}\mathrm{rank}(M^{(i)}).$$
\end{lemma}

\begin{lemma}
\label{lemma2}
If $|E_i|\geq t$ for all $i=1, \ldots, s$ then

$$\Pr\left( \sum_{i=1}^{s} d(M^{(i)}) \geq u\right)\leq K^{s} q^{\frac{-u^2}{s}}$$
where $d(M^{(i)}) = t - \mathrm{rank}(M^{(i)})$ for $i=1, \ldots, s$
and $K$ is a constant depending only on $q$.
\end{lemma}

\begin{proof}
See \cite[Lemma 5]{faugere:2011}.\qed
\end{proof}

\begin{lemma}
\label{lemma3}
Let $u_i = n-(i-1)t$ with $i=\{1, \ldots, s\}$, then
$$\Pr\left( |E_i|< t , \;\middle\vert\; |E_1|\geq t,\ldots, |E_{i-1}|\geq t  \right)\leq e^{-2 \frac{\left(\frac{q-1}{q}u_i-t+1\right)^2}{u_i}}$$
\end{lemma}
\begin{proof}
See \cite[Lemma 6]{faugere:2011}.\qed
\end{proof}

\begin{theorem}
Assume that $st < n$. Then for any function $w(x)$ tending to infinity as $x$ goes to infinity we have
$$\Pr\left(D \geq w(t)\right) = o(1),$$
where $D = st - rank (M)$.
\end{theorem}

\begin{proof}
Note that if $|E_i|\geq t$ for $i\in \{1, \ldots, s\}$ then $D \leq \sum_{i=1}^{s} d(M^{(i)})$.

Let $S_1$ be the event $\sum_{i=1}^{s}d(M^{(i)}) \geq w(t)$ then using Lemma \ref{lemma2}  we have that $\Pr(S_1) = o(1)$. And let $S_2$ be the event of having at least one $E_i$ with $i \in \{1, \ldots, s\}$ such that $|E_i|<t$. Then the probability of the complement of event $S_2$ is given by

$$\Pr\left(\overline{S_2} \right) = \Pr\left(\bigcap_{i=1}^{s} |E_i|\geq t\right) = 
\prod_{i=1}^{s} \Pr\left( |E_1|\geq t,\ldots, |E_i|\geq t \right) = 1 - o(1)$$ by Lemma \ref{lemma3}.

Then we deduce that the sought probability is
$$\Pr\left(D\geq w(t)\right) \leq \Pr\left(S_1 \cup S_2\right) \leq \Pr\left(S_1\right) + \Pr\left(S_2\right) = o(1).$$\qed
\end{proof}

\end{document}